\documentclass[12pt]{article}

\usepackage{fullpage}
\usepackage[english]{babel}
\usepackage[latin1]{inputenc}

\usepackage{amsmath}
\usepackage{amssymb,amsthm}
\usepackage{fancybox}
\usepackage{graphics,graphicx}

\usepackage{stmaryrd}

\newtheorem{theorem}{Theorem}
\newtheorem{lemma}{Lemma}

\begin{document}
\title{\vspace*{-.5cm}\textbf{Optimal Cuts and Bisections
on the Real Line\\[.5ex] in Polynomial Time\\[1.5ex]} }
\author{
  Marek Karpinski\thanks{Research supported partly by DFG grants and the Hausdorff Center grant~EXC59-1. Department of Computer Science, University of Bonn. Email:~{marek@cs.uni-bonn.de}}
  \and
  Andrzej Lingas\thanks{Research supported in part by VR grant 621-2008-4649. Department of Computer Science, Lund University. Email:~{Andrzej.Lingas@cs.lth.se}}
  \and
  Dzmitry Sledneu\thanks{Centre for Mathematical Sciences, Lund
    University. Email:~{Dzmitry.Sledneu@math.lu.se}.}
}
\date{}

\maketitle

\begin{abstract}

The exact complexity of geometric cuts  and bisections is the longstanding open problem including even the dimension one.
In this paper, we resolve this problem for dimension one (the real line) by designing an exact polynomial time algorithm.
Our results depend on a new technique of dealing with metric equalities and their connection to dynamic programming.
The method of our solution could be also of independent interest.
\end{abstract}

\section{Introduction}

The metric MAX-CUT, MAX-BISECTION, MIN-BISECTION and other Partitioning problems were all proved to have \emph{polynomial time approximation schemes}~(PTAS)~\cite{DK98,FK98,AFKK03,FKK04,FKKV05,R10}.
The above problems are known to be NP-hard in exact setting.
The status of those problems for geometric (thus including Euclidean) metrics and this even for dimension one was widely open.

In this paper we resolve the status of those problems for just dimension one by giving a polynomial time algorithm.
Our solution, somewhat surprisingly, involves certain new ideas for applying dynamic programming which could be also of independent interest.

\section{Preliminaries and General Setting}
We shall define our dynamic programming method
in terms of generalized subproblems
on finite multisets of reals
generalizing slightly a geometric metric setting.

For a partition of a finite multiset $P$ of reals
into two multisets $P_1$ and $P_2$, the {\em value of the cut}
is the total length of all intervals on the real line
that have one endpoint in $P_1$ and the other one in $P_2.$

The MAX-CUT problem for $P$ will be now to find a partition
of $P$ into two multisets that maximizes the value of the cut.
If $|P|=n$ and the two multisets are additionally required to be
of cardinality $k$ and $n-k,$ respectively,
then we obtain the $(k,n-k)$ MAX-PARTITION
problem for $P.$ In particular,
if $n$ is even and $k=n/2$ then
we have the MAX-BISECTION problem.
Next, if we replace the requirement of
maximization with that of minimization
then we obtain the $(k,n-k)$ MIN-PARTITION
and MIN-BISECTION problems
for $P$, respectively.

In this paper we study geometric instances of the above problems in dimension one (the real line) which could be rephrased as the problems of partitioning arbitrary finite metric spaces of this dimension.

\section{The Algorithm}

The global idea behind our algorithm is as follows.
We guess how many of the copies of 
the rightmost real are respectively
in the first and second set and move them to
the next to the rightmost real. We also guess
how many copies of the reals in the remaining
part of the multiset are in the first and
second set respectively. Having this information,
we can compute exactly the difference between
the value of optimal solution for the whole
input multiset and that for the multiset resulting
from the movement under the guessed partition proportions.
The difference can be easily evaluated
due to the triangle equality that holds on the real line.
Now, we solve the problem for the transformed
multiset whose elements are copies of
the shrunk set of reals recursively
under the guessed partition proportions, after additionally
guessing the proportions of the partition
for the original copies of the next to the right real.
Eventually, we end up with the trivial multiset
where all elements are copies of the leftmost real.
We eliminate guesses by an enumeration of all possibilities
and choosing the best one. Next, we solve the
resulting subproblems in bottom up fashion
instead the top down to obtain a polynomial time
solution.

Consider a finite multiset $P$ of
reals. We assume that $|P|=n$
and that $P$ consists of copies
of $l\le n$ distinct reals.
For $i=1,...,l,$
let $x_i$ denote the $i$-th smallest real whose copy
is in $P,$ and let $P_i$ denote the sub-multiset
of $P$ consisting of all elements of $P$
which are copies  of reals in $\{ x_1,...,x_{i}\}.$
For convention, we assume $P_0=\emptyset .$

We shall consider a family
of generalized subproblems $S_i(p,q,r,t),$
where $i\in \{1,...,l\},$ $p,q,r,t$ are integers in $\{0,...,n\}$
such that $p+q=|P_{i-1}|$ and $r+t=n- |P_{i-1}|$.
The subproblem $S_i(p,q,r,t)$ is to find for a multiset
that is the union of $P_{i-1}$ with $r+t$ copies of $x_i$
a partition into two multisets 
such that $p$ elements of $P_i$
and $r$ copies of $x_i$ form the first set and the
value of the cut is maximized. The value
of such a maximum cut is denoted by $MAXCUT(S_i(p,q,r,t))$.

\begin{figure}[ht!]
  \centering
  \includegraphics{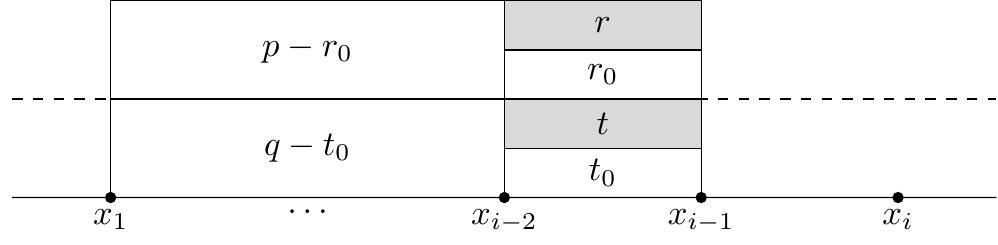}
  \caption{$S_{i-1}(p-r_0,q-t_0,r_0+r,t_0+t).$}
\end{figure}

\begin{figure}[ht!]
  \centering
  \vspace*{.75cm}
  \includegraphics{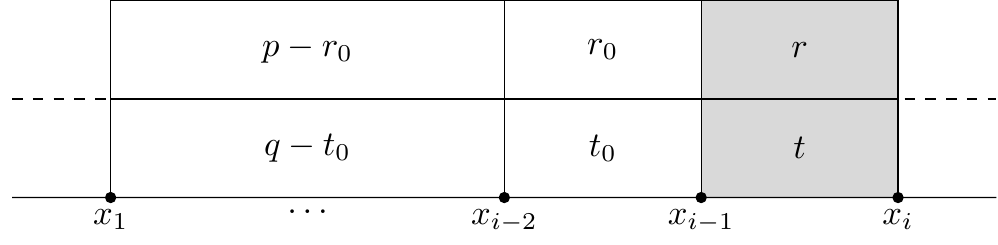}
  \caption{$S_i(p,q,r,t).$}
\end{figure}

\bigskip

\begin{lemma}\label{lem: rec}
For $i\ge 2,$
$MAXCUT(S_i(p,q,r,t))=(x_i-x_{i-1})(pt+qr)+$\newline
$+\max_{r_0\le p \wedge t_0\le q \wedge r_0+t_0=|P_{i-1}|-|P_{i-2}|} 
MAXCUT(S_{i-1}(p-r_0,q-t_0,r_0+r,t_0+t)).$
\end{lemma}
\begin{proof}
For $1\le r_0\le p \wedge 1\le t_0\le q,$
consider an optimal solution to $S_{i-1}(p-r_0,q-t_0,r_0+r,t_0+t),$
where $r_0+t_0=|P_{i-1}|-|P_{i-2}|.$
Let us move $r$ copies of the real $x_{i-1}$ in the first
set of the solution to
the real $x_{i}$ 
and $t$ copies of $x_{i-1}$ in the second
set of the solution to the same real $x_{i}.$
We obtain a feasible solution to $S_i(p,q,r,t)$
whose cut value is 
$(x_i-x_{i-1})(pt+qr)+MAXCUT(S_{i-1}(p-r_0,q-t_0,r_0+r,t_0+t)).$
It follows that 
$MAXCUT(S_i(p,q,r,t))\ge
(x_i-x_{i-1})(pt+qr)+$\newline
$+\max_{r_0\le p \wedge t_0\le q \wedge r_0+t_0=|P_{i-1}|-|P_{i-2}|} 
 MAXCUT(S_{i-1}(p-r_0,q-t_0,r_0+r,t_0+t)).$

Contrary, consider an optimal solution to
$S_i(p,q,r,t).$ Suppose that $r_0$ copies of
$x_{i-1}$ are in the first set of the solution
and $t_0$ copies of $x_{i-1}$ are in the second set
of the solution. 
Note that $r_0+t_0=|P_{i-1}|-|P_{i-2}|$ holds.
Let us move $r$ 
copies of $x_{i}$ in the first
set of the solution to
the real $x_{i-1}$ 
and $t$ copies of $x_{i}$ in the second
set of the solution to the same real $x_{i-1}.$
We obtain a feasible solution to
$S_{i-1}(p-r_0,q-t_0,r_0+r,t_0+t)$ whose cut value is
$MAXCUT(S_i(p,q,r,t))-(x_i-x_{i-1})(pt+qr).$
It follows that 
$MAXCUT(S_i(p,q,r,t))\le (x_i-x_{i-1})(pt+qr)+$\newline
$+\max_{r_0\le p \wedge t_0\le q \wedge r_0+t_0=|P_{i-1}|-|P_{i-2}|} 
MAXCUT(S_{i-1}(p-r_0,q-t_0,r_0+r,t_0+t)).$ %\qed
\end{proof}

\begin{theorem} \label{theo: min}
The 
geometric MAX-CUT problem on the real line as well as the 
geometric MAX-BISECTION problem on the real line
are solvable in $O(n^4)$ time.
\end{theorem}

\begin{proof}
First, we shall show that
the subproblems $S_i(p,q,r,t),$
where $i\in \{1,...,l\},$ $p,q,r,t$ are integers in $\{0,...,n\}$
such that $p+q=|P_{i-1}|$ and $r+t=n- |P_{i-1}|$
are solvable in $O(n^4)$ time.

Since $p+q=|P_{i-1}|$, there are
$O(n)$ choices for the parameters $p,\ q,$
and similarly since $r+t=n- |P_{i-1}|$,
there are $O(n)$ choices for the
parameters $r,\ t.$ It follows that
the total number of considered subproblems
is $O(n^3).$

We can compute the values of $MAXCUT(S_i(p,q,r,t))$
in bottom up
fashion in increasing $i$ order.
If $i=1,$ then $p=0$ and $q=0$ and
$MAXCUT(S_i(p,q,r,t))=0$ holds trivially.
For $i\ge 2,$ we apply Lemma~\ref{lem: rec}
in order to compute
the value of  $MAXCUT(S_i(p,q,r,t))$ in $O(n)$ time
by $r_0+t_0=|P_{i-1}|-|P_{i-2}|$.
The corresponding optimal solutions
can be obtained by backtracking.
 The upper bound $O(n^4)$ follows. 

Let $x_l$ be the largest real whose copy is in
$P,$ and let $m$ be the number
of copies of $x_l$ in $P.$ 
The optimal solution to 
the geometric MAX-CUT problem for $P$
can be found
among the optimal solutions to the $O(n^2)$
subproblems $S_k(p,q,r,t)$, where
$p+q=|P_{k-1}|$ and $r+t=m.$
Furthermore, the optimal solution
to the geometric $(k,n-k)$ MIN-PARTITION problem
for $P$ can be found among
the optimal solutions to the $O(n)$
subproblems $S_k(p,q,r,t)$, where
$p+q=|P_{k-1}|$, $r+t=m,$ $p+r=k$ and $q+t=n-k.$ %\qed
\end{proof}

To solve the $(k,n-k)$ MIN-PARTITION
problem on the real line, we 
consider an analogous family
of generalized subproblems $U_i(p,q,r,t),$
where $i\in \{1,...,l\},$ $p,q,r,t$ 
are integers in $\{0,...,n\}$
such that $p+q=|P_{i-1}|$ and $r+t=n- |P_{i-1}|$.
The subproblem $U_i(p,q,r,t)$ is to find for a multiset
that is the union of $P_{i-1}$ with $r+t$ copies of $x_i$
a partition into two multisets 
such that $p$ elements of $P_i$
and $r$ copies of $x_i$ form the first set and the
value of the cut is minimized. The value
of such a minimum cut is denoted by $MINCUT(S_i(p,q,r,t))$.

Analogously, we obtain the following counterparts
of Lemma \ref{lem: rec}
and Theorem \ref{theo: min} for $(k,n-k)$ MIN-PARTITION
and MIN-BISECTION.

\begin{lemma}
For $i\ge 2,$
$MINCUT(U_i(p,q,r,t))=$
\newline
$(x_i-x_{i-1})(pt+qr)+
\min_{r_0\le p \wedge t_0\le q \wedge r_0+t_0=|P_{i-1}|-|P_{i-2}|}  
MINCUT(U_{i-1}(p-r_0,q-t_0,r_0+r,t_0+t)).$
\end{lemma}

\begin{theorem}
The 
geometric $(k,n-k)$ MIN-PARTITION
problem on the real line, in particular the
geometric MIN-BISECTION problem on the real line,
are analogously solvable 
by dynamic programming in $O(n^4)$ time.
\end{theorem}

\section{Final Remarks}
It remains an open problem whether our method can be generalized to higher dimensions or those problems turn out to be inherently hard.
At stake is the exact computational status of other geometric problems
for which our knowledge is very limited at the moment.

\section{Acknowledgments}
We thank Uri Feige, Ravi Kannan and
Christos Levcopoulos for a number
of interesting discussions
on the subject of this paper.

\end{document}